\documentclass[final]{cvpr}

\usepackage{times}
\usepackage{epsfig}
\usepackage{graphicx}
\usepackage{amsmath}
\usepackage{amssymb}

\usepackage{xr}

\usepackage{bm}
\usepackage{multirow}
\usepackage{subfig}
\usepackage{booktabs}
\usepackage[ruled,vlined]{algorithm2e}
\usepackage{setspace}
\usepackage{enumitem}
\usepackage{nopageno}

\newcommand{\argmin}{\operatornamewithlimits{arg \, min}}

\newtheorem{theorem}{Theorem}

\newtheorem{proof}{Proof}

\usepackage[pagebackref=true,breaklinks=true,colorlinks,bookmarks=false]{hyperref}

\def\confYear{CVPR 2021}

\begin{document}

%%%%%%%%% TITLE
\title{Neighbor2Neighbor: Self-Supervised Denoising from Single Noisy Images}

%\author{Tao Huang\\
%Institute of Statistics and Big Data\\
%Renmin University of China\\
%{\tt\small tao.huang2018@ruc.edu.cn}

% For a paper whose authors are all at the same institution,
% omit the following lines up until the closing ``}''.
% Additional authors and addresses can be added with ``\and'',
% just like the second author.
% To save space, use either the email address or home page, not both
%\and
%Second Author\\
%Institution2\\
%First line of institution2 address\\
%{\tt\small secondauthor@i2.org}
%}

%\maketitle

\renewcommand{\thefootnote}{\fnsymbol{footnote}}

\addtocounter{footnote}{0}
\author{Tao Huang$^{1, 2}$\footnotemark[1], \ Songjiang Li$^{2}$, \ Xu Jia$^{2, 3}$\footnotemark[2], \ Huchuan Lu$^{3}$, \ Jianzhuang Liu$^{2}$ \\
	 {\small $^{1}$Renmin University of China, $^{2}$Noah's Ark Lab, Huawei Technologies, $^{3}$Dalian University of Technology}\\
	{\tt\small tao.huang2018@ruc.edu.cn} 
	\quad {\tt\small jiayushenyang@gmail.com} \\
	{\tt\small lhchuan@dlut.edu.cn} \quad 
	{\tt\small \{songjiang.li, liu.jianzhuang\}@huawei.com}
}

\maketitle

\footnotetext[1]{The work was done in Noah's Ark Lab, Huawei Technologies.}
\footnotetext[2]{Corresponding author}

\renewcommand*{\thefootnote}{\arabic{footnote}}

\setlength{\abovedisplayskip}{5pt}
\setlength{\belowdisplayskip}{0pt}
\setlength{\abovedisplayshortskip}{0pt}
\setlength{\belowdisplayshortskip}{0pt}

%%%%%%%%% ABSTRACT
\begin{abstract}
	\vspace{-5pt}
	In the last few years, image denoising has benefited a lot from the fast development of neural networks.
	However, the requirement of large amounts of noisy-clean image pairs for supervision limits the wide use of these models.
	Although there have been a few attempts in training an image denoising model with only single noisy images, existing self-supervised denoising approaches suffer from inefficient network training, loss of useful information, or dependence on noise modeling.
	In this paper, we present a very simple yet effective method named Neighbor2Neighbor to train an effective image denoising model with only noisy images.
	Firstly, a random neighbor sub-sampler is proposed for the generation of training image pairs. In detail, input and target used to train a network are images sub-sampled from the same noisy image, satisfying the requirement that paired pixels of paired images are neighbors and have very similar appearance with each other.
	Secondly, a denoising network is trained on sub-sampled training pairs generated in the first stage, with a proposed regularizer as additional loss for better performance.
	The proposed Neighbor2Neighbor framework is able to enjoy the progress of state-of-the-art supervised denoising networks in network architecture design. Moreover, it avoids heavy dependence on the assumption of the noise distribution.
	We explain our approach from a theoretical perspective and further validate it through extensive experiments, including synthetic experiments with different noise distributions in sRGB space and real-world experiments on a denoising benchmark dataset in raw-RGB space.
\end{abstract}
\vspace{-10pt}

%-------------------------------------------------------------------------
%%%%%%%%% BODY TEXT
\section{Introduction}

Image denoising is a low-level vision task that is fundamental in computer vision, since noise contamination degrades the visual quality of collected images and may adversely affect subsequent image analysis and processing tasks, such as classification and semantic segmentation \cite{liu2018image}.
Traditional image denoising methods such as BM3D \cite{dabov2007image}, NLM \cite{buades2005non}, and WNNM \cite{gu2014weighted}, use local or non-local structures of an input noisy image. These methods are non-learning-based without the need for ground-truth images.
Recently, convolutional neural networks (CNNs) provide us with powerful tools for image denoising. Numerous CNN-based image denoisers, e.g., DnCNN \cite{zhang2017beyond}, U-Net \cite{geronneberr2015u}, RED \cite{mao2016image}, MemNet \cite{tai2017memnet}, and SGN \cite{gu2019self}, have superior performance over traditional denoisers.
However, CNN-based denoisers depend heavily on a large number of noisy-clean image pairs for training. Unfortunately, collecting large amounts of aligned pairwise noisy-clean training data is extremely challenging and expensive in real-world photography.
Additionally, models trained with synthetic noisy-clean image pairs degrade greatly due to the domain gap between synthetic and real noise.

To mitigate this problem, a series of unsupervised and self-supervised methods that do not require any clean images for training are proposed. These methods require
1) training the network with multiple independent noisy observations per scene \cite{lehtinen2018noise2noise},
2) designing specific blind-spot network structures to learn self-supervised models on only single noisy images \cite{krull2019noise2void, laine2019high, wu2020unpaired}, and making further improvements by using noise models, e.g., Gaussian-Poisson models \cite{laine2019high, wu2020unpaired}, or
3) training the network with noisier-noisy pairs, where the noisier image is derived from the noisy one with synthetic noise added \cite{xu2020noisy, moran2020noisier2noise}.
However, these requirements are not practical in real-world denoising scenarios. 
Firstly, capturing multiple noisy observations per scene remains very challenging, especially for motion scenarios or medical imaging.
Secondly, the relatively low accuracy and heavy computational burden of blind-spot networks greatly limit the application.
Moreover, self-supervised methods with noise model assumptions may work well in synthetic experiments when the noise distribution is known as a prior.
However, these methods degrade sharply when dealing with real-world noisy images where the noise distribution remains unknown.

In this work, we propose Neighbor2Neighbor, a novel self-supervised image denoising framework that overcomes the limitations above.
Our approach consists of a training image pairs generation strategy based on sub-sampling and a self-supervised training scheme with a regularization term.
Specifically, training input and target are generated by random neighbor sub-samplers, where two sub-sampled paired images are extracted from a single noisy image with each element on the same position of the two images being neighbors in the original noisy image.
In this way, if we assume that noise with each pixel is independent conditioned on its pixel value and there is no correlation between noise in different positions, then these two sub-sampled paired noisy images are independent given the ground-truth of the original noisy image. 
Accordingly, inspired by Noise2Noise \cite{lehtinen2018noise2noise}, we use the above training pairs to train a denoising network.
Besides, we develop a regularization term to address the essential difference of pixel ground-truth values between neighbors on the original noisy image.
The proposed self-supervised framework aims at training denoising networks with only single images available, without any modifications to the network structure. Any network that performs well in supervised image denoising tasks can be used in our framework. Moreover, our method does not depend on any noise models either.

To evaluate the proposed Neighbor2Neighbor, a series of experiments on both synthetic and real-world noisy images are conducted. 
The extensive experiments show that our Neighbor2Neighbor outperforms traditional denoisers and existing self-supervised denoising methods learned from only single noisy images. 
The results demonstrate the effectiveness and superiority of the proposed method.

The main contributions of our paper are as follows:
\setlist{nolistsep}
\begin{enumerate}[noitemsep]
	\item We propose a novel self-supervised framework for image denoising, in which any existing denoising networks can be trained without any clean targets, network modifications, or noise model assumptions.
	\item From the theoretical perspective, we provide a sound motivation for the proposed framework. 
	\item Our method performs very favorably against state-of-the-art self-supervised denoising methods especially on real-world datasets, which shows its potential applications in real-world scenarios.
\end{enumerate}

%------------------------------------------------------------------------
\section{Related Work} \label{sec::relatedworks}

\subsection{Supervised Image Denoising}

In the last few years, image denoising based on deep neural networks has been developed rapidly. 
Zhang et al. \cite{zhang2017beyond} proposed DnCNN that combines the convolutional neural network and residual learning for image denoising.
With the supervision of noisy-clean paired images, DnCNN outperforms traditional image denoisers by a large margin.
After that, numerous denoising networks are proposed to further improve the performance \cite{mao2016image, tai2017memnet, zhang2018ffdnet, lefkimmiatis2018universal, plotz2018neural, guo2019toward, gu2019self}.
Nevertheless, these deep denoisers need large amounts of aligned noisy-clean image pairs for training. It is challenging and expensive to collect plenty of training pairs for supervised denoising.
This limits the use of supervised denoisers.

\subsection{Image Denoising with Only Noisy Images}

Image denoising methods using only noisy images can be categorized into two groups: traditional denoisers and deep denoisers.
Traditional denoisers include BM3D \cite{dabov2007image}, NLM \cite{buades2005non}, and WNNM \cite{gu2014weighted}. 
For deep denoisers trained with only a single noisy image, Ulyanov et al. \cite{ulyanov2018deep} proposed deep image prior (DIP), where the image prior is captured from the CNN network rather than specially designed; Self2Self \cite{quan2020self2self} and Noisy-as-Clean \cite{xu2020noisy} are recent works.

Lehtinen et al. \cite{lehtinen2018noise2noise} introduced Noise2Noise to train a deep denoiser with multiple noisy observations of the same scenes.
Subsequently, self-supervised denoising models, including Noise2Void \cite{krull2019noise2void} and Noise2Self \cite{batson2019noise2self}, were proposed to train the networks only with one noisy observation per scene.
Specifically, the carefully designed blind-spot\footnote{We follow the meaning of Krull et al. \cite{krull2019noise2void} that the network prediction for a pixel depends on all input pixels except for the input pixel at its very location.} networks are used to avoid learning the identity transformation. 
Recently, Probabilistic Noise2Void \cite{krull2019probabilistic}, Laine19 \cite{laine2019high}, and Dilated Blind-Spot Network \cite{wu2020unpaired} further introduced explicit noise modeling and probabilistic inference for better performance. Masked convolution \cite{laine2019high} and stacked dilated convolution layers \cite{wu2020unpaired} were introduced for faster training.
Different from blind-spot-based self-supervised methods, in Noisier2Noise \cite{moran2020noisier2noise}, training pairs are prepared by generating synthetic noises from a noise model and adding them to single noisy images.
However, the noise model is hard to specify, especially in real-world scenarios.
Noisy-as-Clean \cite{xu2020noisy} mentioned above shares similar philosophy.

Additionally, Soltanayev and Chun \cite{soltanayev2018training} used Stein's unbiased risk estimator (SURE) to train AWGN denoising models on single noisy images, and Zhussip et al. \cite{zhussip2019extending} extended it to the case of correlated pairs of noisy images. Cha and Moon \cite{cha2019fully} used SURE to fine-tune a supervised denoiser for each test noisy image.
However, SURE-based algorithms are only designed for Gaussian additive noise, and the noise level is required to be known as a prior.

\section{Motivation}\label{sec::Motivation}

In Section \ref{sec::Motivation}, we describe the theoretical framework that motivates our proposed method in Section \ref{sec::neighbor2neighbor}.
The summary of Section \ref{sec::Motivation} is as follows:
In Section \ref{subsec::n2n_revisit}, we revisit the related theory proposed in Noise2Noise, which proves that paired noisy images taken from the same scene can also be used to train denoising models.
Before we extend this theory to the case when only single noisy observation is available, we discuss the case of paired noisy images with slightly different ground-truths in Section \ref{subsec::n2n_similar_gt}, which is very useful for the extension to single noisy images.
Then in Section \ref{subsec::n2n_extension}, we mathematically formulate the underlying theory of training denoising networks where training pairs are generated using image pair samplers, and we further propose a regularizer to solve the problem caused by non-zero $\bm{\varepsilon}$ which is discussed in Section \ref{subsec::n2n_similar_gt}.

\subsection{Noise2Noise Revisit} \label{subsec::n2n_revisit}

Noise2Noise \cite{lehtinen2018noise2noise} is a denoising method trained without the need for ground-truth clean images. This method only requires pairs of independent noisy images of the same scene.
Given two independent noisy observations named $\mathbf{y}$ and $\mathbf{z}$ of the same ground-truth image $\mathbf{x}$, Noise2Noise tries to minimize the following loss in terms of $\theta$,
\begin{align} \label{equ::n2n}
	\underset{\theta}{\arg\min} ~ \mathbb{E}_{\mathbf{x}, \mathbf{y}, \mathbf{z}} \left\lVert f_\theta(\mathbf{y}) - \mathbf{z}\right\rVert_2^2,
\end{align}
where $f_\theta$ is the denoising network parameterized by $\theta$.
Minimizing Equation \eqref{equ::n2n} yields the same solution as the supervised training with the $\ell_2$-loss.
For detailed discussions, refer to Section 2 of \cite{lehtinen2018noise2noise} and Section 3.1 of \cite{zhussip2019extending}.

\subsection{Paired Images with Similar Ground Truths} \label{subsec::n2n_similar_gt}

Noise2Noise \cite{lehtinen2018noise2noise} mitigates the need of clean images.
However, capturing multiple noisy observations of a scene remains a very challenging problem.
The ground-truths of two noisy observations are difficult to be the same due to occlusion, motion, and lighting variation.
Thus, we propose to extend the Equation \eqref{equ::n2n} to the case where the gap between the underlying clean images $\bm{\varepsilon} := \mathbb{E}_{\mathbf{z}|\mathbf{x}}(\mathbf{z}) - \mathbb{E}_{\mathbf{y}|\mathbf{x}}(\mathbf{y}) \neq \mathbf{0}$.
\begin{theorem}\label{theorem::epsilon}
	Let $\mathbf{y}$ and $\mathbf{z}$ be two independent noisy images conditioned on $\mathbf{x}$, and assume that there exists an $\bm{\varepsilon} \neq \mathbf{0}$ such that $\mathbb{E}_{\mathbf{y}|\mathbf{x}}(\mathbf{y})=\mathbf{x}$ and $\mathbb{E}_{\mathbf{z}|\mathbf{x}}(\mathbf{z}) = \mathbf{x}+\bm{\varepsilon}$. 
	Let the variance of $\mathbf{z}$ be $\bm{\sigma}_\mathbf{z}^2$.
	Then it holds that
	\begin{equation} \label{equ::n2n_epsilon}
		\begin{aligned}
			\mathbb{E}_{\mathbf{x}, \mathbf{y}} \left\lVert f_{\theta}(\mathbf{y})-\mathbf{x} \right\rVert_2^2
			&= \mathbb{E}_{\mathbf{x}, \mathbf{y}, \mathbf{z}} \left\lVert f_{\theta}(\mathbf{y})-\mathbf{z} \right\rVert_2^2-\bm{\sigma}_\mathbf{z}^2 \\
			&+ 2\bm{\varepsilon} \mathbb{E}_{\mathbf{x}, \mathbf{y}}(f_{\theta}(\mathbf{y})-\mathbf{x}).  			
		\end{aligned} 
	\end{equation}
\end{theorem}
The proof is given in the supplementary material.
Theorem \ref{theorem::epsilon} states that when the gap $\bm{\varepsilon} \neq \mathbf{0}$, since $ \mathbb{E}_{\mathbf{x}, \mathbf{y}}(f_{\theta}(\mathbf{y})-\mathbf{x}) \not\equiv 0$,  optimizing $\mathbb{E}_{\mathbf{x}, \mathbf{y}, \mathbf{z}}\|f_{\theta}(\mathbf{y})-\mathbf{z}\|_2^2$ does not yield the same solution as the supervised training loss $\mathbb{E}_{\mathbf{x}, \mathbf{y}}\|f_{\theta}(\mathbf{y})-\mathbf{x}\|_2^2$.
Fortunately, if $\bm{\varepsilon} \to \mathbf{0}$, which means the gap is sufficiently small, $2\bm{\varepsilon} \mathbb{E}_{\mathbf{x}, \mathbf{y}}(f_{\theta}(\mathbf{y})-\mathbf{x}) \to \mathbf{0}$, so the network trained with noisy image pair $(\mathbf{y}, \mathbf{z})$ works as a reasonable approximate solution to the supervised training network.
Note that when $\bm{\varepsilon} = \mathbf{0}$, since $\bm{\sigma}_\mathbf{z}^2$ is a constant, minimizing both sides of Equation \eqref{equ::n2n_epsilon} results in $\argmin_{\theta} \ \mathbb{E}_{\mathbf{x},\mathbf{y},\mathbf{z}} \left\lVert f_\theta(\mathbf{y}) - \mathbf{z} \right\rVert_2^2$
, which is the objective of Noise2Noise.

\subsection{Extension to Single Noisy Images} \label{subsec::n2n_extension}

Inspired by Noise2Noise where training pairs are independent noisy image pairs of the same scene, we go a step further and propose to generate independent training pairs from single noisy images $\mathbf{y}$ by sampling.

To be specific, an image pair sampler $G = (g_1, g_2)$ is used to generate a noisy image pair $(g_1(\mathbf{y}), g_2(\mathbf{y}))$ from a single noisy image $\mathbf{y}$.
The contents of two sampled images $(g_1(\mathbf{y}), g_2(\mathbf{y}))$ are closely resembled.
Similar to Equation \eqref{equ::n2n}, we try to adopt the sampled image pair as two noisy observations, which becomes:
\begin{align} \label{equ::pn2n}
	\argmin_{\theta} ~ \mathbb{E}_{\mathbf{x}, \mathbf{y}} \left\lVert f_\theta(g_1(\mathbf{y})) - g_2(\mathbf{y}) \right\rVert^2.
\end{align}
Different from Noise2Noise, the ground-truths of two sampled noisy images $(g_1(\mathbf{y}), g_2(\mathbf{y}))$ differ, i.e., $\bm{\varepsilon} = \mathbb{E}_{\mathbf{y}|\mathbf{x}}(g_2(\mathbf{y})) - \mathbb{E}_{\mathbf{y}|\mathbf{x}}(g_1(\mathbf{y})) \neq \mathbf{0}$.
According to Theorem \ref{theorem::epsilon}, directly applying Equation \eqref{equ::pn2n} is not appropriate and leads to over-smoothing.
Thus, we consider the non-zero gap $\bm{\varepsilon}$.

Considering the optimal (ideal) denoiser $f_\theta^*$ that is trained with clean images and the $\ell_2$-loss, given $\mathbf{x}$, it satisfies that $f_\theta^*(\mathbf{y}) = \mathbf{x}$ and $f_\theta^*(g_\ell(\mathbf{y})) = g_\ell(\mathbf{x})$, for $\ell \in \{1, 2\}$.
Thus, the following holds with the optimal network $f_\theta^*$: 

\vspace{-12pt}
\begin{small}
	\begin{align}  
		& \mathbb{E}_{\mathbf{y}|\mathbf{x}} ~ \{ f_{\theta}^*(g_1(\mathbf{y}))-g_2(\mathbf{y})-\left(g_1(f_{\theta}^*(\mathbf{y}))-g_2(f_{\theta}^*(\mathbf{y}))\right) \}
		\label{equ::constraint1} \\
		& = g_1(\mathbf{x}) - \mathbb{E}_{\mathbf{y}|\mathbf{x}} \{g_2(\mathbf{y})\} - (g_1(\mathbf{x}) - g_2(\mathbf{x}))
		\nonumber \\ 
		& = g_2(\mathbf{x}) - \mathbb{E}_{\mathbf{y}|\mathbf{x}} \{g_2(\mathbf{y})\} = 0.    \nonumber %\label{equ::constraint2}
	\end{align}	
\end{small} 

\vspace{-15pt}
{
	\noindent With the last two terms in Equation \eqref{equ::constraint1}, we consider the gap between the ground truths of the training image pair.
	If the gap is zero, the subtraction of the last two terms in Equation \eqref{equ::constraint1} vanishes, and Equation \eqref{equ::constraint1} becomes a special case of Noise2Noise paired training in Equation \eqref{equ::n2n}.
	However, if the gap is non-zero, these two terms serve as a correction of the ground truth gap between the first two terms in Equation \eqref{equ::constraint1}, 
	forcing \eqref{equ::constraint1} to be zero.
}

Therefore, Equation \eqref{equ::constraint1} provides a constraint that is satisfied when a denoiser $f_\theta$ is the ideal one $f_\theta^*$.
To exploit this (ideal) constraint, rather than directly optimizing Equation \eqref{equ::pn2n}, we consider the following constrained optimization problem:

\vspace{-10pt}
\small
\begin{equation}\label{eq::regular_opt}
	\begin{aligned}
		&\min_{\theta} ~ \mathbb{E}_{\mathbf{y}|\mathbf{x}} \left\lVert f_{\theta}(g_1(\mathbf{y})) - g_2(\mathbf{y}) \right\rVert_2^2, \ \text{ s.t.}  \\
		& \mathbb{E}_{\mathbf{y}|\mathbf{x}} \{ f_{\theta}(g_1(\mathbf{y})) - g_2(\mathbf{y}) -  g_1(f_{\theta}(\mathbf{y})) + g_2(f_{\theta}(\mathbf{y})) \} = 0.		
	\end{aligned}
\end{equation}
\normalsize

{
	\noindent
	With the equation $\mathbb{E}_{\mathbf{x}, \mathbf{y}} = \mathbb{E}_{\mathbf{x}} \mathbb{E}_{\mathbf{y}| \mathbf{x}}$, we further reformulate it as the following regularized optimization problem:
}

\vspace{-9pt}
\small
\begin{equation}
	\begin{aligned}
		&\min_{\theta} ~ \mathbb{E}_{\mathbf{x}, \mathbf{y}} \left\lVert f_{\theta}(g_1(\mathbf{y})) - g_2(\mathbf{y}) \right\rVert_2^2 \\
		+& \gamma \mathbb{E}_{\mathbf{x}, \mathbf{y}}  \left\lVert f_{\theta}(g_1(\mathbf{y})) - g_2(\mathbf{y}) - g_1(f_{\theta}(\mathbf{y})) + g_2(f_{\theta}(\mathbf{y})) \right\rVert_2^2.		
	\end{aligned}
\end{equation}
\normalsize

\begin{figure*}
	\centering
	\captionsetup[subfigure]{}
	\subfloat[][Training] { 
		\includegraphics[height=0.37\linewidth]{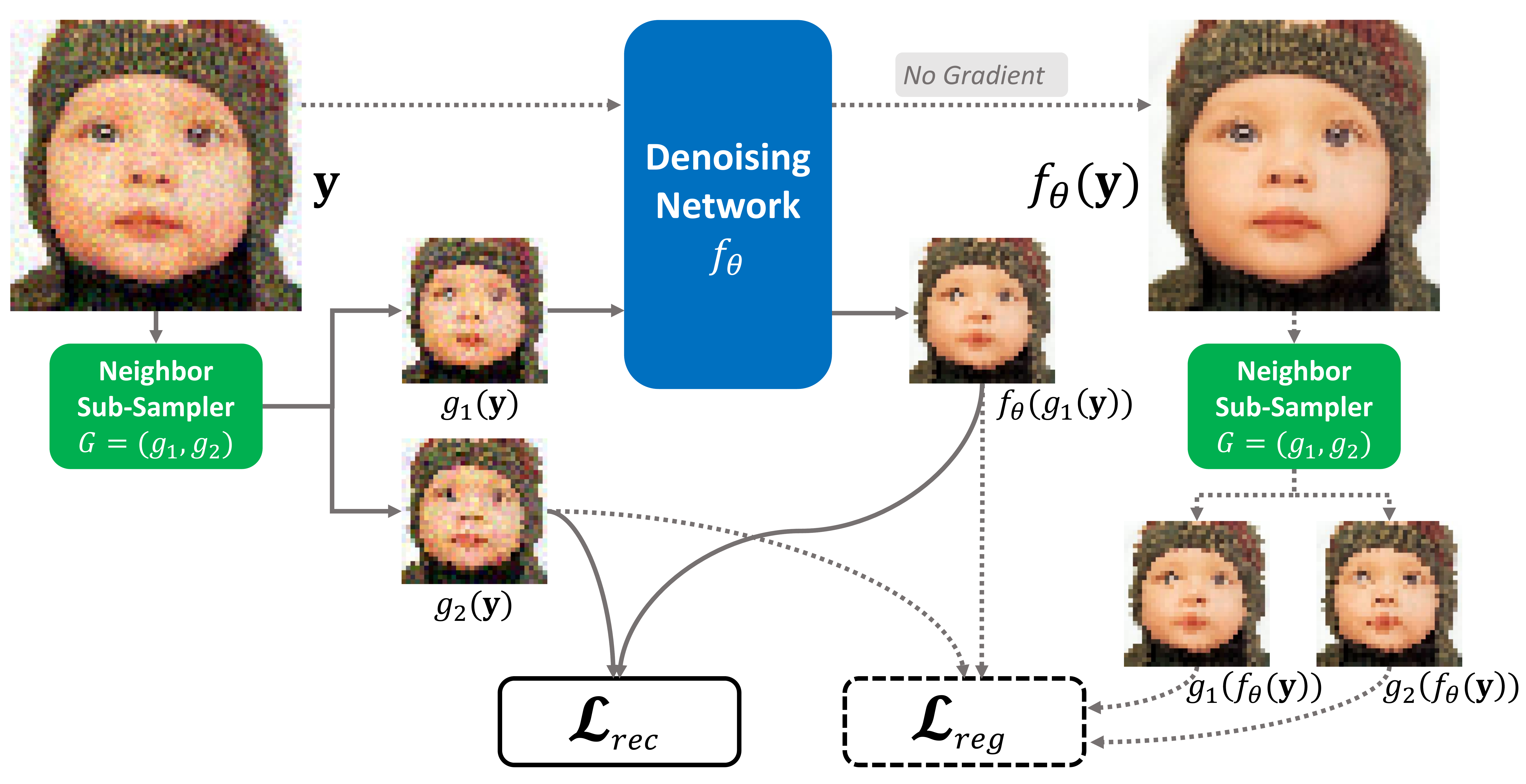} 
		\label{framework:train}
	}
	\subfloat[][Inference] {
		\includegraphics[height=0.37\linewidth]{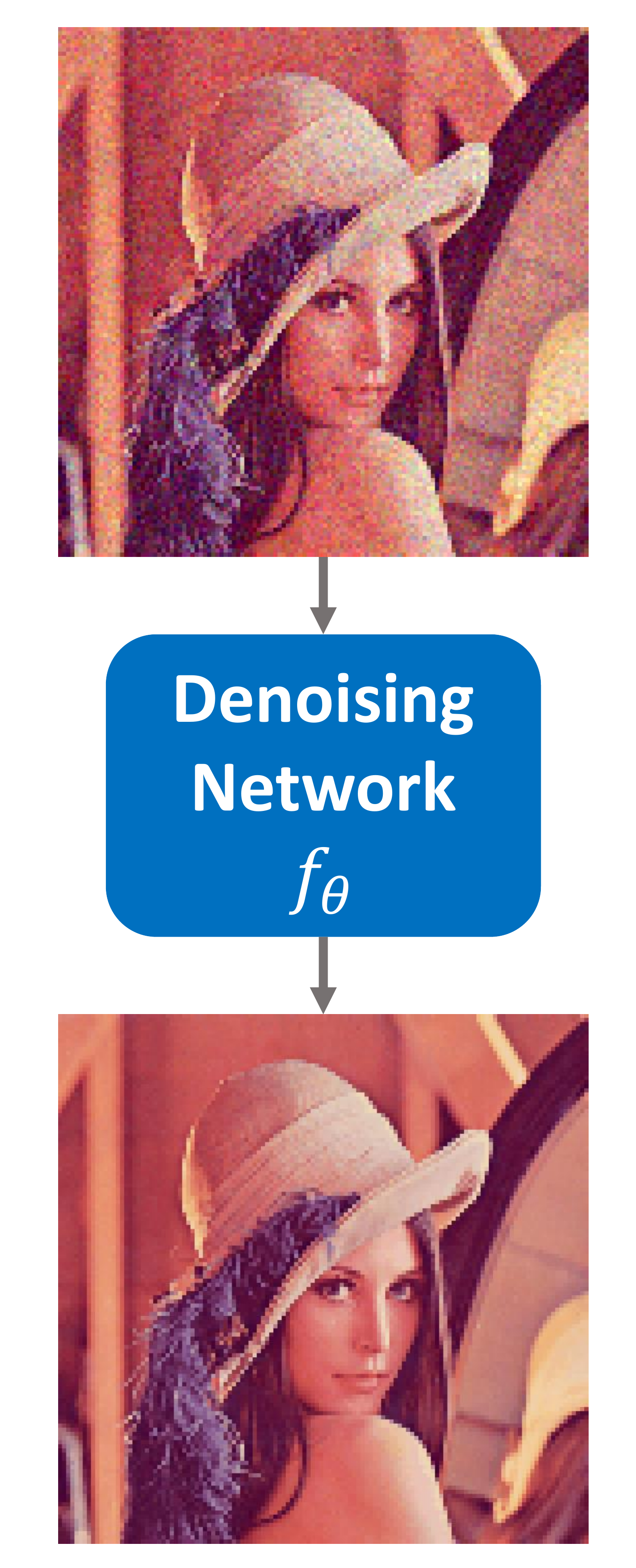}
		\label{framework:test}
	}
	\vspace{-6pt}
	\caption{Overview of our proposed Neighbor2Neighbor framework.
		(a) Complete view of the training scheme.
		A pair of sub-sampled images $(g_1(\mathbf{y}), g_2(\mathbf{y}))$ are generated from a noisy image $\mathbf{y}$ with a neighbor sub-sampler $G=(g_1, g_2)$.	
		The denoising network takes $g_1(\mathbf{y})$ and $g_2(\mathbf{y})$ as input and target respectively.
		The regularized loss $\mathcal{L}$ includes the following two terms:
		On the left side, the reconstruction term $\mathcal{L}_{rec}$ is computed between the network output and the noisy target.
		On the right side, the regularization term $\mathcal{L}_{reg}$ is further added, considering the essential difference of the ground-truth pixel values between the sub-sampled noisy image pair.
		It should be mentioned that the neighbor sub-sampler $G$ (\textcolor{green}{green}) that appears twice represents the same neighbor sub-sampler.
		(b) Inference using the trained denoising network.
		Best viewed in color.
	}
	\label{fig:framework}
\end{figure*}

\section{Proposed Method} \label{sec::neighbor2neighbor}

In this section, based on our motivation in Section \ref{sec::Motivation}, we propose Neighbor2Neighbor, a self-supervised framework to train CNN denoisers from single observation of noisy images. 
The proposed training scheme consists of two parts.
The first is to generate pairs of noisy images by using random \textit{neighbor sub-samplers}. 
For the second part, while the sub-sampled image pairs are used for self-supervised training, we further introduce a regularized loss to address the non-zero ground-truth gap between the paired sub-sampled noisy images.
The regularized loss consists of a reconstruction term and a regularization term.
An overview of our proposed Neighbor2Neighbor framework including training and inference is shown in Figure \ref{fig:framework}.

\subsection{Generation of Training Image Pairs} \label{subsec::noisypairgenerate}

Firstly, we introduce a neighbor sub-sampler to generate noisy image pairs $(g_1(\mathbf{y}), g_2(\mathbf{y}))$ from single noisy images $\mathbf{y}$ for training, satisfying the following assumptions discussed in Section \ref{subsec::n2n_extension}:
1) the sub-sampled paired noisy images $(g_1(\mathbf{y}), g_2(\mathbf{y}))$ are conditionally independent given the ground-truth $\mathbf{x}$ of $\mathbf{y}$;
2) the gap between the underlying ground-truth images of $g_1(\mathbf{y})$ and $g_2(\mathbf{y})$ is small.

\begin{figure}[t]
	\begin{center}
		\includegraphics[width=0.9\linewidth]{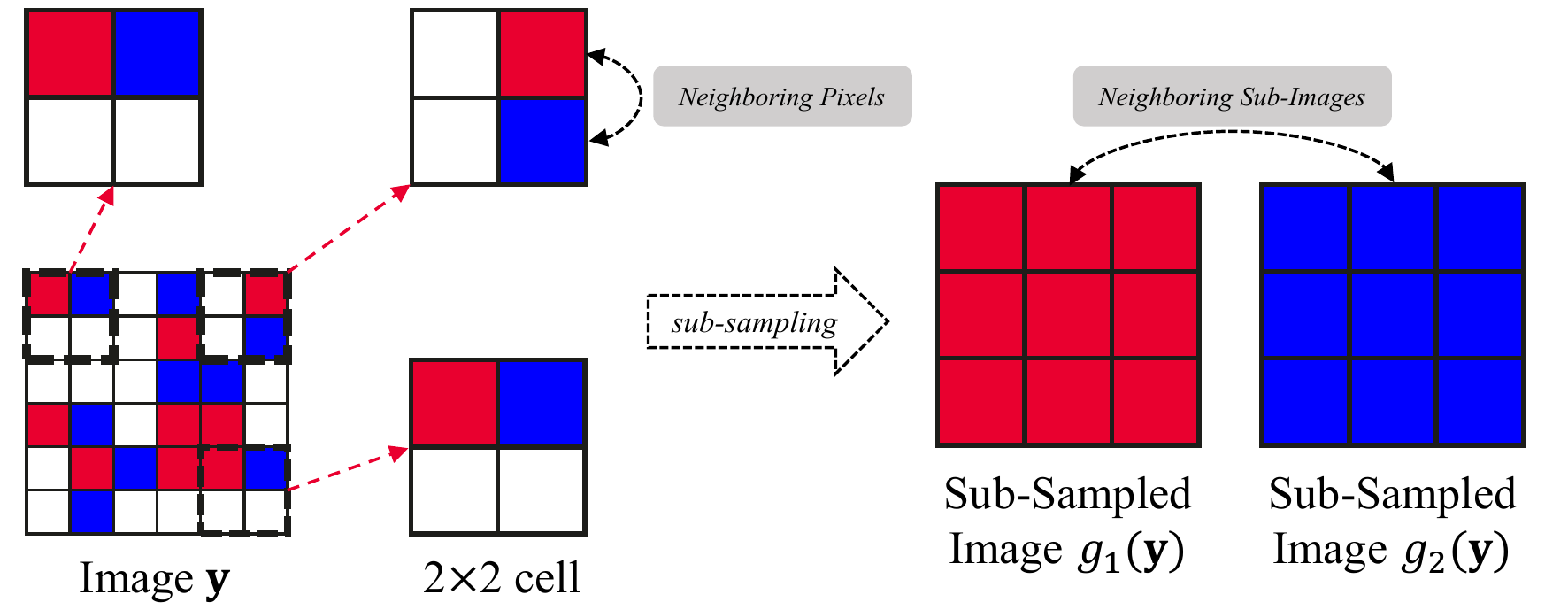}
	\end{center}
	\vspace{-15pt}
	\caption{Example of generating an image pair with a neighbor sub-sampler $G=(g_1, g_2)$. 
		Here, $k = 2$ and in each $2 \times 2$ cell, two neighboring pixels are randomly chosen,  filled in \textcolor{red}{red} and \textcolor{blue}{blue} respectively.
		The pixel filled in \textcolor{red}{red} is taken as a pixel of a sub-sampled image $g_1(\mathbf{y})$, and the other pixel filled in \textcolor{blue}{blue} is taken as a pixel of another sub-sampled image $g_2(\mathbf{y})$.
		The sub-sampled paired images $(g_1(\mathbf{y}), g_2(\mathbf{y}))$ are shown as the red patch and the blue patch on the right.
		Best viewed in color.}
	\vspace{-10pt}
	\label{fig:downsampling}
\end{figure}

The diagram of an image pair generation with a neighbor sub-sampler is shown in Figure \ref{fig:downsampling}.
Denote an image as $\mathbf{y}$ with width $W$ and height $H$.
The details of a neighbor sub-sampler $G = (g_1, g_2)$ are described below:
\begin{enumerate}
	\item The image $\mathbf{y}$ is divided into $\lfloor W/k \rfloor \times \lfloor H/k \rfloor$ cells with each of size $k \times k$. Empirically, we set $k = 2$.
	\item For the $i$-th row and $j$-th column cell, two neighboring locations are randomly selected. They are respectively taken as the $(i,j)$-th element of the sub-sampler $G = (g_1, g_2)$.
	\item For all the $\lfloor W/k \rfloor \times \lfloor H/k \rfloor$ cells, repeat step 2. Then the neighbor sub-sampler $G=(g_1, g_2)$ is generated. 
	Given the image $\mathbf{y}$, two sub-sampled images $(g_1(\mathbf{y}), g_2(\mathbf{y}))$ with size $\lfloor W/k \rfloor \times \lfloor H/k \rfloor$ are derived.
\end{enumerate}

In this way, we can use different random neighbor sub-samplers to generated noisy training image pairs from single noisy image.
The ground-truths of the paired images $(g_1(\mathbf{y}), g_2(\mathbf{y}))$ are similar, since paired pixels of $(g_1(\mathbf{y}), g_2(\mathbf{y}))$ are neighbors in the original noisy image $\mathbf{y}$.
Moreover, the requirement of independence of $(g_1(\mathbf{y}), g_2(\mathbf{y}))$ given $\mathbf{x}$ is satisfied, if we further assume that the noisy image $\mathbf{y}$ is conditionally pixel-wise independent given the ground-truth $\mathbf{x}$.

\subsection{Self-Supervised Training with a Regularizer} \label{subsec::self-supervised_training}

Since we have generated training image pairs from single noisy images with our proposed neighbor sub-samplers in Section \ref{subsec::noisypairgenerate}, here we will introduce our self-supervised training strategy with noisy training inputs and targets.

Given a pair of sub-sampled images $(g_1(\mathbf{y}), g_2(\mathbf{y}))$ from a noisy image $\mathbf{y}$, we use the following regularized loss developed in Section \ref{subsec::n2n_extension} to train the denoising network:
\begin{align}\label{equ::psen2n_rerm}
 \mathcal{L} & = \mathcal{L}_{rec} + \gamma \cdot \mathcal{L}_{reg} \nonumber\\
& = \left\lVert f_\theta(g_1(\mathbf{y})) - g_2(\mathbf{y})\right\rVert_2^2 \\
& + \gamma \cdot \left\lVert f_\theta(g_1(\mathbf{y})) - g_2(\mathbf{y}) - (g_1(f_\theta(\mathbf{y})) - g_2(f_\theta(\mathbf{y}))) \right\rVert_2^2, \nonumber	
\end{align}
where $f_\theta$ is a denoising network with arbitrary network design, and $\gamma$ is a hyper-parameter controlling the strength of the regularization term.
To stabilize learning, we stop the gradients of $g_1(f_\theta(\mathbf{y}))$ and $g_2(f_\theta(\mathbf{y}))$, and gradually increase $\gamma$ to the specified value in the training process.
The training framework is shown in Algorithm \ref{alg::training}.
\vspace{-10pt}
\begin{algorithm}
	\SetAlgoLined
	\KwIn{A set of noisy images $Y = \{\mathbf{y}_i \}_{i=1}^n$ ; \newline
		Denoising network $f_\theta$; \newline
		Hyper-parameter $\gamma$.}
	\While{not converged }{
		Sample a noisy image $\mathbf{y} \in Y$\;
		Generate a random neighbor sub-sampler $G=(g_1, g_2)$\;
		Derive a pair of sub-sampled images $(g_1(\mathbf{y}), g_2(\mathbf{y}))$, where $g_1(\mathbf{y})$ is the network input, and $g_2(\mathbf{y})$ is the network target\;       
		For the network input $g_1(\mathbf{y})$, derive the denoised image $f_\theta(g_1(\mathbf{y}))$\;
		Calculate $\mathcal{L}_{rec} =  \lVert f_\theta(g_1(\mathbf{y})) - g_2(\mathbf{y})\rVert^2$\;
		For the original noisy image $\mathbf{y}$, derive the denoised image $f_\theta(\mathbf{y})$ with no gradients\;
		Use the same neighbor sub-sampler $G$ to derive the pair $(g_1(f_\theta(\mathbf{y})), g_2(f_\theta(\mathbf{y})))$\;
		Calculate $\mathcal{L}_{reg} =\lVert f_\theta(g_1(\mathbf{y})) - g_2(\mathbf{y}) - (g_1(f_\theta(\mathbf{y})) - g_2(f_\theta(\mathbf{y}))) \rVert^2$\;
		Update the denoising network $f_\theta$ by minimizing the objective $ \mathcal{L}_{rec} + \gamma \cdot \mathcal{L}_{reg}$. 
	}
	\caption{Neighbor2Neighbor}
	\label{alg::training}
\end{algorithm}

\vspace{-10pt}

\section{Experiments} \label{sec::experiments}

In this section, we first describe the implementation details.
Then, to evaluate the effectiveness, the proposed method is compared with state-of-the-art denoising methods trained with single noisy images for 1) synthetic Gaussian or Poisson denoising in sRGB space, and 2) real-world noisy image denoising in raw-RGB space.
Furthermore, ablation studies are conducted to analyze the effectiveness of the proposed sub-sampler and regularizer.

\subsection{Implementation Details} \label{subsec::implementation-details}

% Training details of our approach
\noindent\textbf{Training Details.} For better comparisons, we follow \cite{laine2019high} and use a modified version of U-Net \cite{geronneberr2015u} architecture with three $1 \times 1$ convolution layers added at the end of the network.
We use a batch size of 4 for training and use Adam optimizer with an initial learning rate of $0.0003$ for synthetic denoising experiments in sRGB space and $0.0001$ for real-world denoising experiments in raw-RGB space. The number of training epochs is 100, and the learning rate is half-decayed per 20 epochs.
As for the hyper-parameter $\gamma$ used to control the strength of the regularization term, we set $\gamma = 2$ in the synthetic experiments and $\gamma = 1$ in the real-world experiments empirically.
All experiments are conducted on a server with Python 3.6.4, PyTorch 1.3 \cite{paszke2019pytorch} and Nvidia Tesla V100 GPUs. 
Code is available at {\urlstyle{same} \url{https://github.com/TaoHuang2018/Neighbor2Neighbor}}.
We also implement our algorithm on Mindspore\footnote{\urlstyle{same} \url{https://www.mindspore.cn/}}.

% Dataset preparation
\noindent\textbf{Datasets for Synthetic Experiments.} 
For synthetic experiments in sRGB space, we use 50k images from ImageNet \cite{deng2009imagenet} validation dataset as the source of clean images. Similar to \cite{laine2019high}, we only select clean images whose sizes are between $256 \times 256$ and $512 \times 512$ pixels, and then we randomly crop $256 \times 256$ patches for training.
We consider the following four types of synthetic noise distributions:
(1) Gaussian noise with a fixed level $\sigma = 25$, (2) Gaussian noise with varied noise levels $\sigma \in [5, 50]$, (3) Poisson noise with a fixed level $\lambda = 30$, and (4) Poisson noise with varied noise levels $\lambda \in [5, 50]$.
It should be mentioned that these $\sigma$ values correspond to image color intensities in $[0, 255]$, while these $\lambda$ values correspond to the intensities in $[0, 1]$.
We use Kodak \cite{franzen1999kodak}, BSD300 \cite{martin2001database}, and Set14 \cite{zeyde2010single} image sets for testing. 
% In order to obtain reliable average PSNRs, we replicate the testing evaluation on Kodak, BSD300 and Set14 by 10, 3, and 20 times respectively.

\noindent\textbf{Datasets for Real-World Experiments.}
For real-world experiments in raw-RGB space, we consider SIDD \cite{abdelhamed2018high} dataset, a real-world denoising
benchmark,  which is collected using five smartphone cameras under 10 static scenes.
We use SIDD Medium Dataset in RAW format for training, and use SIDD Validation and Benchmark Datasets for validation and testing respectively.

\subsection{Comparisons with State-of-the-Arts} \label{subsec::exp}

% Comparisons
\noindent\textbf{Compared Methods.} 
We compare our Neighbor2Neighbor against two baseline methods (supervised denoising (N2C) and Noise2Noise (N2N) \cite{lehtinen2018noise2noise}), one traditional denoiser (BM3D \cite{dabov2007image}), and six self-supervised denoisers (Deep Image Prior (DIP) \cite{ulyanov2018deep}, Noise2Void (N2V) \cite{krull2019noise2void}, Self2Self \cite{quan2020self2self}, Laine19 \cite{laine2019high}, Noisier2Noise \cite{moran2020noisier2noise} and DBSN \cite{wu2020unpaired}).

\noindent\textbf{Details of Synthetic Experiments.}
In these experiments, 
1) for the two baseline methods (N2C and N2N) and Laine19, we use the pre-trained network weights provided in \cite{laine2019high}, to keep the same network architecture as ours; 
2) for BM3D, we use CBM3D with the parameter $\sigma$ estimated by the method in \cite{chen2015efficient} to denoise Gaussian noise, and use Anscombe transform \cite{makitalo2010optimal} to denoise Poisson noise;
3) for DIP, Self2Self, N2V, and DBSN, we use the authors' implementation, and for Noisier2Noise, we re-implement the authors' design with $\alpha=1$. 
Besides, Laine19 with probabilistic post-processing (posterior mean estimation) is denoted as Laine19-pme, and the method without post-processing is denoted as Laine19-mu.

\begin{table} 
	\setlength\tabcolsep{2.5pt}
	\begin{center}
		\footnotesize
		\begin{tabular}{clccc}
			\toprule
			Noise Type & Method & KODAK & BSD300 & SET14 \\
			\midrule
			\multirow{10}{*}{\shortstack{Gaussian \\ $\sigma = 25$}} 
			& Baseline, N2C \cite{geronneberr2015u} & 32.43/0.884 & 31.05/0.879 & 31.40/0.869 \\  
			& Baseline, N2N \cite{lehtinen2018noise2noise} & 32.41/0.884 & 31.04/0.878 & 31.37/0.868 \\ \cline{2-5}
			& CBM3D \cite{dabov2007image} & 31.87/0.868 & 30.48/0.861 & 30.88/0.854 \\ 
			& DIP \cite{ulyanov2018deep} & 27.20/0.720 & 26.38/0.708 & 27.16/0.758 \\
			& Self2Self \cite{quan2020self2self} & 31.28/0.864 & 29.86/0.849 & 30.08/0.839 \\  
			& N2V \cite{krull2019noise2void} & 30.32/0.821 & 29.34/0.824 & 28.84/0.802 \\  
			& Laine19-mu \cite{laine2019high} & 30.62/0.840 & 28.62/0.803 & 29.93/0.830 \\  
			& Laine19-pme \cite{laine2019high} & \textbf{32.40}/\textbf{0.883} & \textbf{30.99}/\textbf{0.877} & \textbf{31.36}/\textbf{0.866} \\ 
			& Noisier2Noise \cite{moran2020noisier2noise} & 30.70/0.845 & 29.32/0.833 & 29.64/0.832 \\   
			& DBSN \cite{wu2020unpaired} & 31.64/0.856 & 29.80/0.839 & 30.63/0.846 \\ 
			& Ours & \underline{32.08}/\underline{0.879} & \underline{30.79}/\underline{0.873} & \underline{31.09}/\underline{0.864} \\
			\midrule
			\multirow{9}{*}{\shortstack{Gaussian \\ $\sigma \in [5, 50]$}} 
			& Baseline, N2C \cite{geronneberr2015u} & 32.51/0.875 & 31.07/0.866 & 31.41/0.863 \\  
			& Baseline, N2N \cite{lehtinen2018noise2noise} & 32.50/0.875 & 31.07/0.866 & 31.39/0.863 \\ \cline{2-5}  
			& CBM3D \cite{dabov2007image} & 32.02/\underline{0.860} & 30.56/\underline{0.847} & 30.94/0.849 \\  
			& DIP \cite{ulyanov2018deep} & 26.97/0.713 & 25.89/0.687 & 26.61/0.738 \\ 
			& Self2Self \cite{quan2020self2self} & 31.37/0.860 & 29.87/0.841 & 29.97/0.849 \\  
			& N2V \cite{krull2019noise2void} & 30.44/0.806 & 29.31/0.801 & 29.01/0.792 \\ 
			& Laine19-mu \cite{laine2019high} & 30.52/0.833 & 28.43/0.794 & 29.71/0.822 \\ 
			& Laine19-pme \cite{laine2019high} & \textbf{32.40}/\textbf{0.870} & \textbf{30.95}/\textbf{0.861} & \textbf{31.21}/\underline{0.855} \\  
			& DBSN \cite{wu2020unpaired} & 30.38/0.826 & 28.34/0.788 & 29.49/0.814 \\  
			& Ours & \underline{32.10}/\textbf{0.870} & \underline{30.73}/\textbf{0.861} & \underline{31.05}/\textbf{0.858}\\
			\midrule
			\multirow{9}{*}{\shortstack{Poisson \\ $\lambda = 30$}} 
			& Baseline, N2C \cite{geronneberr2015u} & 31.78/0.876 & 30.36/0.868 & 30.57/0.858 \\  
			& Baseline, N2N \cite{lehtinen2018noise2noise} & 31.77/0.876 & 30.35/0.868 & 30.56/0.857 \\   \cline{2-5}
			& Anscombe \cite{makitalo2010optimal} & 30.53/0.856 & 29.18/0.842 & 29.44/0.837 \\   
			& DIP \cite{ulyanov2018deep} & 27.01/0.716 & 26.07/0.698 & 26.58/0.739 \\ 
			& Self2Self \cite{quan2020self2self} & 30.31/0.857 & 28.93/0.840 & 28.84/0.839 \\  
			& N2V \cite{krull2019noise2void} & 28.90/0.788 & 28.46/0.798 & 27.73/0.774 \\  
			& Laine19-mu \cite{laine2019high} & 30.19/0.833 & 28.25/0.794 & 29.35/0.820 \\ 
			& Laine19-pme \cite{laine2019high} & \textbf{31.67}/\textbf{0.874} & \textbf{30.25}/\textbf{0.866} & \textbf{30.47}/\textbf{0.855} \\  
			& DBSN \cite{wu2020unpaired} & 30.07/0.827 & 28.19/0.790 & 29.16/0.814 \\ 
			& Ours & \underline{31.44}/\underline{0.870} & \underline{30.10}/\underline{0.863} & \underline{30.29}/\underline{0.853} \\
			\hline
			\multirow{9}{*}{\shortstack{Poisson \\ $\lambda \in [5, 50]$}} 
			& Baseline, N2C \cite{geronneberr2015u} & 31.19/0.861 & 29.79/0.848 & 30.02/0.842 \\  
			& Baseline, N2N \cite{lehtinen2018noise2noise} & 31.18/0.861 & 29.78/0.848 & 30.02/0.842 \\    \cline{2-5}
			& Anscombe \cite{makitalo2010optimal} & 29.40/0.836 & 28.22/0.815  & 28.51/\underline{0.817}  \\  
			& DIP \cite{ulyanov2018deep} & 26.56/0.710 & 25.44/0.671 & 25.72/0.683 \\  
			& Self2Self \cite{quan2020self2self} & 29.06/0.834 & 28.15/0.817 & 28.83/0.841 \\  
			& N2V \cite{krull2019noise2void} & 28.78/0.758 & 27.92/0.766 & 27.43/0.745 \\  
			& Laine19-mu \cite{laine2019high} & 29.76/0.820 & 27.89/0.778 & \underline{28.94}/0.808 \\  
			& Laine19-pme \cite{laine2019high} & \textbf{30.88}/\underline{0.850} & \textbf{29.57}/\underline{0.841} & 28.65/0.785 \\  
			& DBSN \cite{wu2020unpaired} & 29.60/0.811 & 27.81/0.771 & 28.72/0.800 \\  
			& Ours & \underline{30.86}/\textbf{0.855} & \underline{29.54}/\textbf{0.843} & \textbf{29.79}/\textbf{0.838} \\
			\bottomrule
		\end{tabular}
	\end{center}
	\footnotesize
	\vspace{-15pt}
	\caption{Quantitative comparison (PSNR(dB)/SSIM) of different methods for Gaussian noise and Poisson noise.
		For each noise type, the highest PSNR(dB)/SSIM among the denoising methods trained without clean images is marked in \textbf{bold} while the second highest is \underline{underlined}.}
	\vspace{-5pt}
	\label{tab::synthetic}
\end{table}

\begin{figure*}
	\centering
	\captionsetup[subfigure]{labelformat=empty}
	\includegraphics[width=1.0\textwidth]{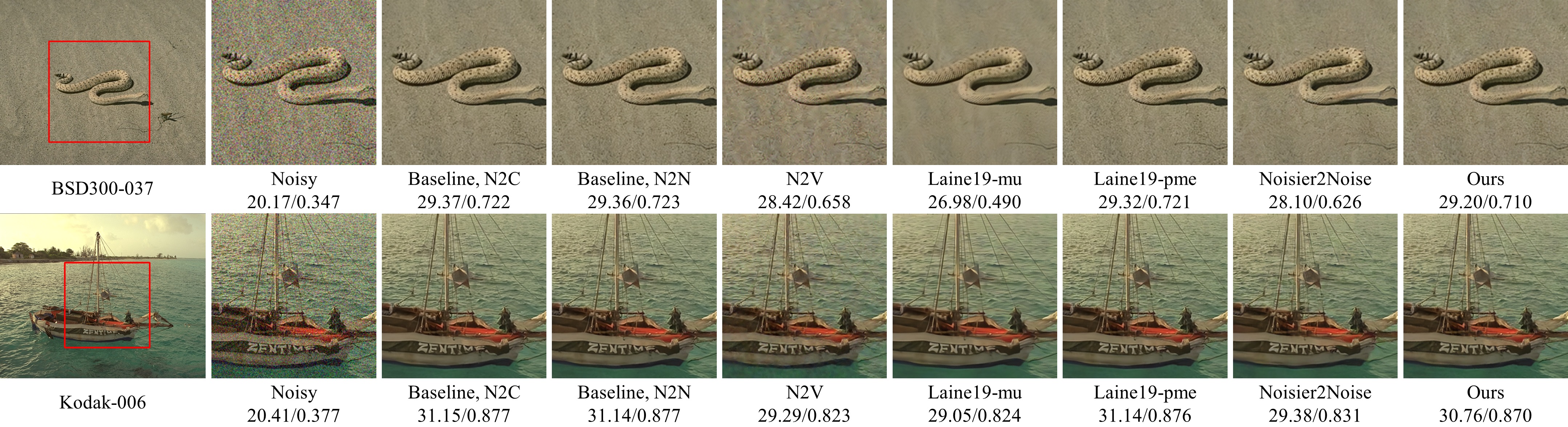}
	\vspace{-20pt}
	\caption{Visual comparison of our method against other competing methods in the setting of Gaussian $\sigma = 25$.
		The quantitative PSNR(dB)/SSIM results are listed underneath the images.
		Best viewed in color.}
	\vspace{-10pt}
	\label{fig::vis-gauss25}
\end{figure*}

\begin{figure*}
	\centering
	\captionsetup[subfigure]{labelformat=empty}
	\includegraphics[width=1.0\textwidth]{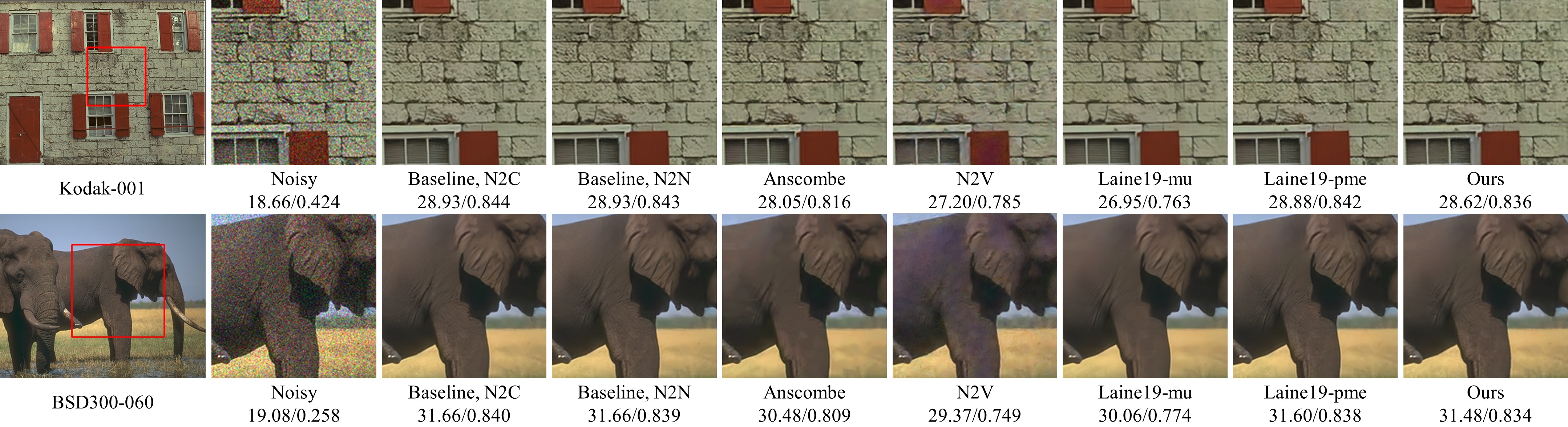}
	\vspace{-20pt}
	\caption{Visual comparison of our method against other competing methods in the setting of Poisson $\lambda = 30$.
		The quantitative PSNR(dB)/SSIM results are listed underneath the images.
		Best viewed in color.}
	\vspace{-10pt}
	\label{fig::vis-poisson30}
\end{figure*}

\noindent\textbf{Results of Synthetic Experiments.} 
The quantitative comparisons including the PSNR and SSIM results on the Gaussian noise and the Poisson noise are shown in Table \ref{tab::synthetic}.
It can be seen that our Neighbor2Neighbor performs not only better than traditional denoisers like BM3D (CBM3D for Gaussian noise and Anscombe transform for Poisson noise), but also better than six self-supervised denoising methods DIP, Self2Self, N2V, Noisier2Noise, Laine19-mu, and DBSN.
This indicates that without the consideration of explicit noise modeling in the synthetic experiments, our method is superior among existing self-supervised denoising frameworks.
Moreover, compared with Laine19-pme which needs to model noise distribution as a prior, our approach also achieves a comparable performance, especially in the cases with varied noise levels.
One possible reason is that it is difficult to estimate the noise parameter $\gamma$ or $\lambda$ for each image in such a case, which harms the performance of probabilistic post-processing.
Therefore, if a noise distribution is complex and unknown, methods with probabilistic post-processing may severely degrade due to the difficulty of noise modeling.
Further real-world experiments show that explicit noise modeling may not be suitable in real-world image denoising.
In Figures \ref{fig::vis-gauss25} and \ref{fig::vis-poisson30}, we provide the qualitative comparisons of the denoised images for Gaussian noise with noise level $\sigma = 25$ and Poisson noise with noise level $\lambda = 30$.
One can see that our method produces competitive visual quality among traditional and self-supervised deep denoisers.

\noindent\textbf{Details of Real-world Experiments.}  
In the experiments with real-world noisy images in raw-RGB space, for all baseline methods and compared methods, we use the authors' implementation and train on the SIDD Medium dataset, except the results of the network trained on only CycleISP-generated synthetic image pairs that are depicted in the supplementary material of \cite{zamir2020cycleisp}.
The denoising details are
1) for BM3D, we split a noisy image into four sub-images according to the Bayer pattern, denoise them individually, and then recombine the denoised sub-images into a single denoised image;
2) for network-based methods, we use the packed 4-channel raw images as the network input and unpack the network output to obtain denoised raw images.
Note that the network architectures of the network-based methods are listed in Table \ref{tab::sidd}.
And Laine19 models parenthesized with `Gaussian' or `Poisson' mean that Gaussian distribution or Poisson distribution is used to model the real-world noise explicitly.

\noindent\textbf{Results of Real-world Experiments.} 
Table \ref{tab::sidd} lists the PSNR/SSIM results of different denoising algorithms on SIDD Validation and SIDD Benchmark Datasets in raw-RGB space.
The results on SIDD Benchmark Dataset are provided by the online server \cite{sidd}. 
Our approach not only performs better than traditional denoising algorithms but also consistently outperforms the self-supervised deep denoisers under the same network architecture. Besides, compared with the RRGs network trained with only CycleISP-generated synthetic image pairs, our approach also performs better.
Notice that the methods with probabilistic post-processing are inferior to the same methods without the post-processing. This is because explicit noise modeling is difficult in real-world photography, and simply modeling the noise distribution with Gaussian/Heteroscedastic Gaussian/Poisson distribution is not enough. Therefore, the methods with this post-processing cannot generalize well to real-world denoising.
Furthermore, the performance of our approach can be further improved if we replace UNet \cite{geronneberr2015u} with a more advanced denoising architecture composed of multiple RRGs in \cite{zamir2020cycleisp}.
This indicates that the proposed Neighbor2Neighbor framework is able to enjoy the progress of state-of-the-art  image denoising networks in network architecture design.
Figure \ref{fig::vis-sidd_benchmark} presents some visual comparisons of our model against other methods, showing the effectiveness of our approach.

\begin{table}
	\setlength\tabcolsep{3pt}
	\begin{center}
		\footnotesize
		\begin{tabular}{llcc}
			\toprule
			\multirow{2}{*}{Methods} & \multirow{2}{*}{Network} & \multicolumn{1}{c}{SIDD} & \multicolumn{1}{c}{SIDD}  \\
			&  & \multicolumn{1}{c}{Benchmark} & \multicolumn{1}{c}{Validation} \\ 
			\midrule
			Baseline, N2C \cite{geronneberr2015u} & U-Net \cite{geronneberr2015u} & 50.60/0.991 & 51.19/0.991 \\ 
			Baseline, N2N \cite{lehtinen2018noise2noise} & U-Net \cite{geronneberr2015u}  & 50.62/0.991 & 51.21/0.991 \\ 
			\midrule
			BM3D \cite{dabov2007image} & - & 48.60/0.986 & 48.92/0.986 \\ 
			BM3D$^*$ \cite{dabov2007image} & - & 45.52/0.980 & \;\;  -  \;\;  / \;\;  -  \;\; \\ 
			N2V \cite{krull2019noise2void} & U-Net \cite{geronneberr2015u}  & 48.01/0.983 & 48.55/0.984 \\  
			Laine19-mu (Gaussian) \cite{laine2019high} & U-Net \cite{geronneberr2015u}  & 49.82/0.989 & 50.44/0.990 \\ 
			Laine19-pme (Gaussian) \cite{laine2019high} & U-Net \cite{geronneberr2015u}  & 42.17/0.935 & 42.87/0.939 \\ 
			Laine19-mu (Poisson) \cite{laine2019high} & U-Net \cite{geronneberr2015u}  & 50.28/0.989 & 50.89/0.990 \\ 
			Laine19-pme (Poisson) \cite{laine2019high} & U-Net \cite{geronneberr2015u}  & 48.46/0.984 & 48.98/0.985 \\ 
			DBSN \cite{wu2020unpaired} & DBSN \cite{wu2020unpaired}  & 49.56/0.987 & 50.13/0.988 \\ 
			CycleISP (synthetic) \cite{zamir2020cycleisp}  & RRGs \cite{zamir2020cycleisp} & \;\;  -  \;\;  / \;\;  -  \;\; & 50.45/ \;\;  -  \;\;  \\  
			% \midrule
			Ours  & U-Net \cite{geronneberr2015u} & \textbf{50.47/0.990} & \textbf{51.06/0.991} \\ 
			Ours  & RRGs \cite{zamir2020cycleisp}  & \textbf{50.76/0.991}  & \textbf{51.39/0.991} \\ 
			\bottomrule
		\end{tabular}
	\end{center}
	\vspace{-12pt}
	\footnotesize
	\caption{Quantitative comparisons (PSNR(dB)/SSIM) on SIDD benchmark and validation datasets in raw-RGB space.
		The best PSNR(dB)/SSIM results among denoising methods without the need for clean images are marked in \textbf{bold}.
		$*$ denotes that results are obtained from the website of SIDD Benchmark \cite{sidd}. }
	\vspace{-12pt}
	\label{tab::sidd}
\end{table}

\begin{figure*}
	\centering
	\captionsetup[subfigure]{labelformat=empty}
	\subfloat[][] {
		\includegraphics[width=1.0\linewidth]{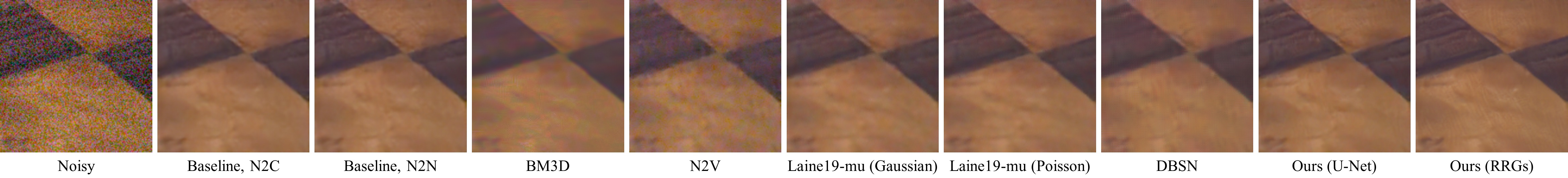}
	} \\
	\vspace{-20pt}
	\subfloat[][] {
		\includegraphics[width=1.0\linewidth]{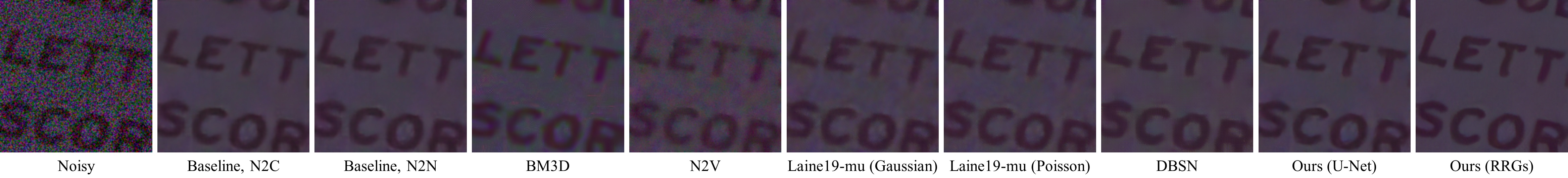}
	}
	\vspace{-20pt}
	\caption{Visual comparison of our method against other methods on SIDD Benchmark.
		All images are converted from raw-RGB space to sRGB space by the ISP provided by SIDD\protect\footnotemark ~for visualization.
		Best viewed in color.}
	\vspace{-10pt}
	\label{fig::vis-sidd_benchmark}
\end{figure*}

\subsection{Ablation Study} \label{subsec::ablation_study}

Here, we conduct ablation studies to analyze the influence of the regularization term and the sampling strategy.

\noindent\textbf{Influence of Regularization Term.}
The hyper-parameter $\gamma$ in Equation \eqref{equ::psen2n_rerm} is used to control the strength of the regularization term. Table \ref{tab::ablation_gamma} lists the performance of Neighbor2Neighbor under different $\gamma$ values on Kodak dataset. 
From Table \ref{tab::ablation_gamma} and Figure \ref{fig::vis-ablation_gamma}, we have the following observations: 
1) When $\gamma = 0$, that is, the regularization term is removed, the denoising performances severely suffer from the gap between the underlying ground truths of the sub-sampled paired images. 
The corresponding denoised image is over-smoothing and lacks of detailed information. 
2) As $\gamma$ increases, the denoised image becomes sharper, and 
3) When $\gamma$ is too large, the regularization term dominates the loss. The denoising performance becomes worse as a great deal of noise remains. 
To this end, the regularization term acts as a controller between smoothness and noisiness. A moderate $\gamma$ value is selected to obtain both clean and sharp results. 
In this paper, we use $\gamma=2$ for synthetic experiments, and $\gamma=1$ for real-world experiments.
\begin{table} 
	\setlength\tabcolsep{2pt}
	\begin{center}
		\footnotesize
		\begin{tabular}{lcccc}
			\toprule
			Noise Type & $\gamma=0$ & $\gamma=2$ & $\gamma=8$ & $\gamma=20$ \\
			\midrule
			Gaussian $\sigma = 25$ & 31.77/0.874 & \textbf{32.08}/\textbf{0.879}  & 32.02/0.878 & 31.95/0.874 \\			
			Gaussian $\sigma \in [5, 50]$ & 31.67/0.866 & \textbf{32.10}/\textbf{0.870}  & 31.99/0.865 & 31.87/0.861 \\
			Poisson $\lambda = 30$ & 31.21/0.866 & \textbf{31.44}/\textbf{0.870}  & 31.38/\textbf{0.870} & 31.21/0.864 \\
			Poisson $\lambda \in [5, 50]$ & 30.67/0.853 & \textbf{30.86}/\textbf{0.855}  & 30.74/0.851 & 30.58/0.846 \\
			\bottomrule
		\end{tabular}
	\end{center}
	\footnotesize
	\vspace{-12pt}
	\caption{Ablation on different weights ($\gamma$ values) of the regularizer. PSNR (dB) and SSIM results are evaluated on the Kodak dataset.}
	\label{tab::ablation_gamma}
\end{table}
\begin{figure}
	\centering
	\includegraphics[width=1.0\linewidth]{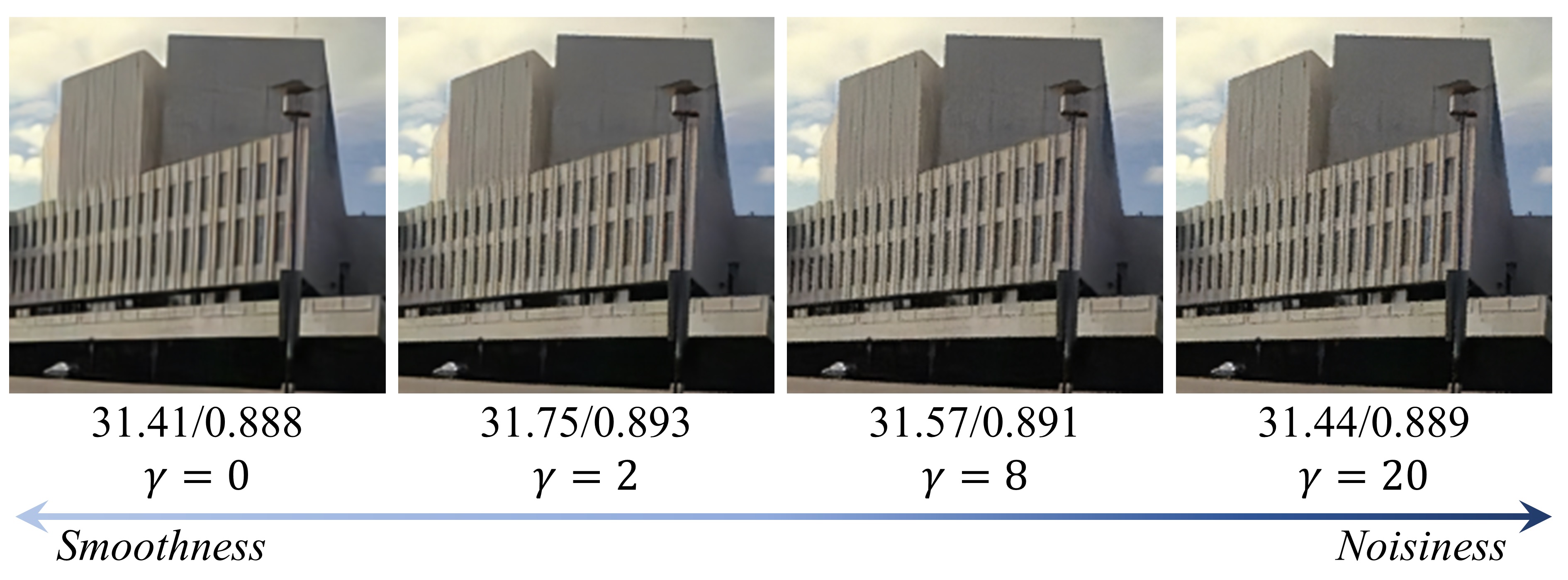}
	\caption{Visual comparison of Neighbor2Neighbor under different $\gamma$ values in the setting of Gaussian $\sigma=25$.
		The quantitative PSNR(dB)/SSIM results are listed underneath the images.
		Best viewed in color.}
	\label{fig::vis-ablation_gamma}
\end{figure}

\footnotetext{\urlstyle{same} \url{https://github.com/AbdoKamel/simple-camera-pipeline}}

\noindent\textbf{Influence of Sampling Strategy.}
To evaluate the effectiveness of our sampling strategy, we compare our random neighbor sub-sampler in Section \ref{subsec::noisypairgenerate} to a naive fix-location sub-sampler (fix-location), where $k^2$ sub-sampled images from one noisy image are generated.
In each sub-sampled image, all pixels are from the same location of all the $k \times k$ cells. 
Then two sub-sampled images are randomly selected from these $k^2$ images. 
Details are provided in the supplementary materials.
See Table \ref{tab::ablation_sampling} for the quantitative comparison with $\gamma = 2$.
\begin{table} 
	\setlength\tabcolsep{4pt}
	\begin{center}
		\footnotesize
		\begin{tabular}{clcc}
			\toprule
			Dataset & Noise Type & Fix-location & Random \\
			\midrule
			\multirow{4}{*}{Kodak} & Gaussian $\sigma = 25$ & 32.05/0.876 & \textbf{32.08/0.879}	\\
			& Gaussian $\sigma \in [5, 50]$ & 32.07/0.867 & \textbf{32.10/0.870} \\
			& Poisson $\lambda = 30$ & 31.42/0.869  & \textbf{31.44/0.870} \\
			& Poisson $\lambda \in [5, 50]$ & 30.83/0.853 & \textbf{30.86/0.855}	\\
			\midrule
			SIDD Validation & real-world noise & 50.92/\textbf{0.991} & \textbf{51.06/0.991} \\
			\bottomrule
		\end{tabular}
	\end{center}
	\footnotesize
	\caption{Ablation on different sampling strategies. PSNR (dB) and SSIM results are evaluated on the Kodak dataset and the SIDD validation dataset.}
	\label{tab::ablation_sampling}
	
\end{table}
The comparison of `Random' v.s. `Fix-location' shows the important role of randomness in the sampling strategy. The fix-location sub-sampler is a special case of our method with much less randomness, and therefore the proposed random neighbor sub-sampler achieves better performance.

\section{Conclusion} \label{sec::conclusion}

We propose Neighbor2Neighbor, a novel self-supervised framework for image denoising which puts an end to the need for noisy-clean pairs, multiple noisy observations, blind-spot networks and explicit noise modeling.
Our approach successfully solves single images denoising by generating sub-sampled paired images with random neighbor sub-samplers as training image pairs and using the self-supervised training scheme with a regularization term.
The extensive experiments have shown the effectiveness and superiority of the proposed Neighbor2Neighbor over existing methods.  
For future work, we would like to extend the proposed method to the case of spatially-correlated noise and extremely dark images.

{\small
\bibliographystyle{ieee_fullname}
\bibliography{egbib}
}

\newpage
\quad
\newpage

\appendix
\setcounter{table}{0}  
\setcounter{figure}{0}  
\setcounter{theorem}{0}  
\renewcommand{\thesection}{\Alph{section}}
\renewcommand\thefigure{A.\arabic{figure}} 

\section{Proof of Theorem \ref*{theorem::epsilon}} \label{sup::proofs}

\begin{theorem}
	Let $\mathbf{y}$ and $\mathbf{z}$ be two independent noisy images conditioned on $\mathbf{x}$, and assume that there exists an $\bm{\varepsilon} \neq \mathbf{0}$ such that $\mathbb{E}_{\mathbf{y}|\mathbf{x}}(\mathbf{y})=\mathbf{x}$ and $\mathbb{E}_{\mathbf{z}|\mathbf{x}}(\mathbf{z}) = \mathbf{x}+\bm{\varepsilon}$. 
	Let the variance of $\mathbf{z}$ be $\bm{\sigma}_\mathbf{z}^2$.
	Then it holds that
	\begin{equation} 
		\begin{aligned}
			\mathbb{E}_{\mathbf{x}, \mathbf{y}} \left\lVert f_{\theta}(\mathbf{y})-\mathbf{x} \right\rVert_2^2
			&= \mathbb{E}_{\mathbf{x}, \mathbf{y}, \mathbf{z}} \left\lVert f_{\theta}(\mathbf{y})-\mathbf{z} \right\rVert_2^2-\bm{\sigma}_\mathbf{z}^2 \\
			&+ 2\bm{\varepsilon} \mathbb{E}_{\mathbf{x}, \mathbf{y}}(f_{\theta}(\mathbf{y})-\mathbf{x}). 		
		\end{aligned}
	\end{equation}
\end{theorem}

\begin{proof}
	First, similar to the derivation in Section 2 of the supplementary materials of \cite{zhussip2019extending}, we have
	\begin{equation}
		\begin{aligned}
			\mathbb{E}_{\mathbf{y}|\mathbf{x}} \left\lVert f_{\theta}(\mathbf{y})-\mathbf{x} \right\rVert_2^2
			& = \mathbb{E}_{\mathbf{y},\mathbf{z}|\mathbf{x}}  \left\lVert f_{\theta}(\mathbf{y})-\mathbf{z}+\mathbf{z}-\mathbf{x} \right\rVert_2^2 \\
			& = \mathbb{E}_{\mathbf{y}, \mathbf{z}|\mathbf{x}} \left\lVert f_{\theta}(\mathbf{y})-\mathbf{z} \right\rVert_2^2 + \mathbb{E}_{\mathbf{z}|\mathbf{x}} \left\lVert \mathbf{z}-\mathbf{x} \right\rVert_2^2 \\
			& + 2\mathbb{E}_{\mathbf{y},\mathbf{z}|\mathbf{x}}(f_{\theta}(\mathbf{y})-\mathbf{z})^{\top}(\mathbf{z}-\mathbf{x})\\
			& = \mathbb{E}_{\mathbf{y}, \mathbf{z}|\mathbf{x}} \left\lVert f_{\theta}(\mathbf{y})-\mathbf{z} \right\rVert_2^2 + \bm{\sigma}_\mathbf{z}^2 \\
			& + 2\mathbb{E}_{\mathbf{y},\mathbf{z}|\mathbf{x}}(f_{\theta}(\mathbf{y})-\mathbf{x}+\mathbf{x}-\mathbf{z})^{\top}(\mathbf{z}-\mathbf{x})\\
			& = \mathbb{E}_{\mathbf{y}, \mathbf{z}|\mathbf{x}} \left\lVert f_{\theta}(\mathbf{y})-\mathbf{z} \right\rVert_2^2 + \bm{\sigma}_\mathbf{z}^2 \\
			& + 2\mathbb{E}_{\mathbf{y},\mathbf{z}|\mathbf{x}}(f_{\theta}(\mathbf{y})-\mathbf{x})^{\top}(\mathbf{z}-\mathbf{x}) \\
			& + 2\mathbb{E}_{\mathbf{z}|\mathbf{x}}(\mathbf{x}-\mathbf{z})^{\top}(\mathbf{z}-\mathbf{x})\\
			& = \mathbb{E}_{\mathbf{y}, \mathbf{z}|\mathbf{x}} \left\lVert f_{\theta}(\mathbf{y})-\mathbf{z} \right\rVert_2^2 - \bm{\sigma}_\mathbf{z}^2 \\
			& + 2\mathbb{E}_{\mathbf{y},\mathbf{z}|\mathbf{x}}(f_{\theta}(\mathbf{y})-\mathbf{x})^{\top}(\mathbf{z}-\mathbf{x}).
		\end{aligned}		
	\end{equation}
	Due to the independence between $\mathbf{y}$ and $\mathbf{z}$ given $\mathbf{x}$, it holds that
	\begin{equation}
		\begin{aligned}
			\mathbb{E}_{\mathbf{y}|\mathbf{x}} \left\lVert f_{\theta}(\mathbf{y})-\mathbf{x} \right\rVert_2^2 
			&= \mathbb{E}_{\mathbf{y}, \mathbf{z}|\mathbf{x}} \left\lVert f_{\theta}(\mathbf{y})-\mathbf{z} \right\rVert_2^2 - \bm{\sigma}_\mathbf{z}^2 \\
			& + 2\mathbb{E}_{\mathbf{y}|\mathbf{x}}(f_{\theta}(\mathbf{y})-\mathbf{x})^{\top}\mathbb{E}_{\mathbf{z}|\mathbf{x}}(\mathbf{z}-\mathbf{x}) \\
			& = \mathbb{E}_{\mathbf{y}, \mathbf{z}|\mathbf{x}} \left\lVert f_{\theta}(\mathbf{y})-\mathbf{z} \right\rVert_2^2 - \bm{\sigma}_\mathbf{z}^2 \\
			& + 2\bm{\varepsilon}\mathbb{E}_{\mathbf{y}|\mathbf{x}}(f_{\theta}(\mathbf{y})-\mathbf{x}).
		\end{aligned}		
	\end{equation}
	Since $\mathbb{E}_{\mathbf{x}, \mathbf{y}} = \mathbb{E}_{\mathbf{x}} \mathbb{E}_{\mathbf{y}| \mathbf{x}}$, we further have
	\begin{equation}
		\begin{aligned}
			\mathbb{E}_{\mathbf{x}, \mathbf{y}} \left\lVert f_{\theta}(\mathbf{y})-\mathbf{x} \right\rVert_2^2
			& = \mathbb{E}_{\mathbf{x}, \mathbf{y}, \mathbf{z}} \left\lVert f_{\theta}(\mathbf{y})-\mathbf{z} \right\rVert_2^2-\bm{\sigma}_\mathbf{z}^2 \\
			& + 2\bm{\varepsilon} \mathbb{E}_{\mathbf{x}, \mathbf{y}}(f_{\theta}(\mathbf{y})-\mathbf{x}). 
		\end{aligned} 		
	\end{equation}
	
\end{proof}

\newpage

\section{Details of Fix-Location Sampling Strategy}

Here, we give an illustrative example to describe the details of fix-location sampling strategy.
The fix-location sampler randomly generates a pair of sub-sampled images from $k^2$ to-be-chosen sub-sampled images. In each sub-sampled image, all pixels are from the same location of all the $k \times k$ cells. 
In Figure \ref{fig:downsampling_app}, $k = 2$, and locations chosen for four sub-sampled images (red, blue, yellow, and green pixels) are totally the same in each $2 \times 2$ cell on the left.
Consequently, four sub-sampled images are generated, filled in red, blue, yellow, and green in the middle. 
Then, the sub-sampled paired images $(g_1(\mathbf{y}), g_2(\mathbf{y}))$ are randomly selected which is shown as the blue patch and the green patch on the right.

\begin{figure}[h]
	\begin{center}
		\includegraphics[width=\linewidth]{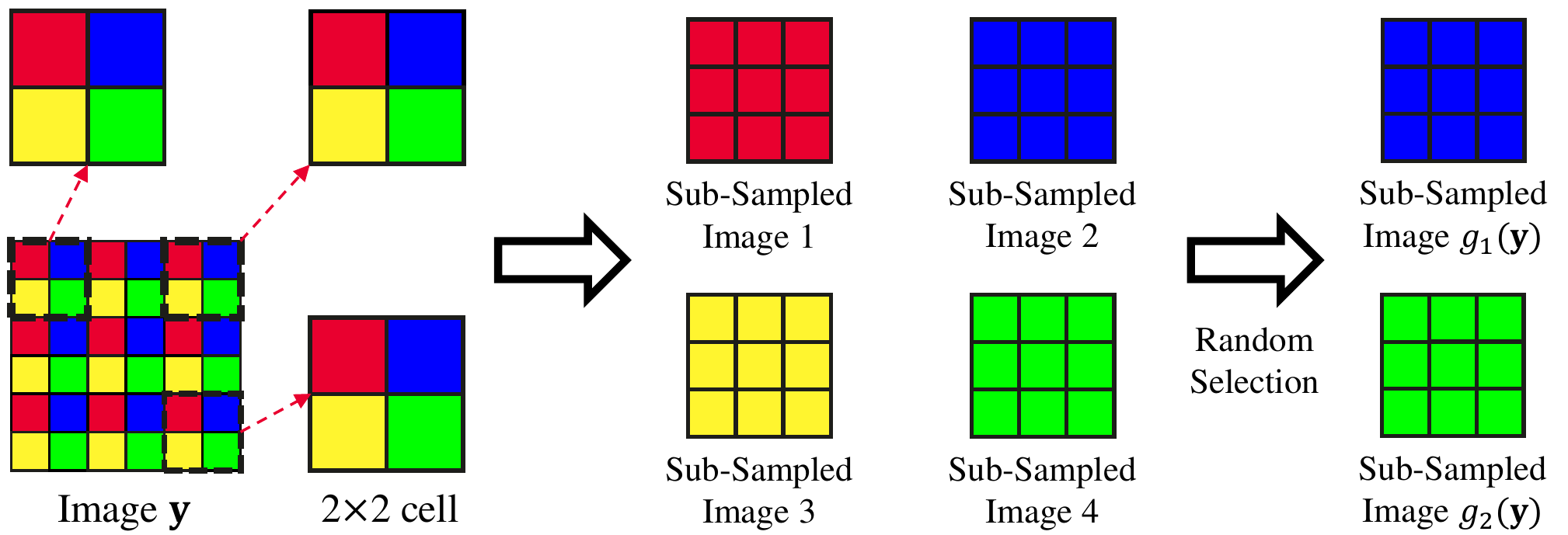}
	\end{center}
	\vspace{-10pt}
	\caption{Example of image pair generation with a fix-location sub-sampler $G=(g_1, g_2)$. 
		Best viewed in color.}
	\vspace{-10pt}
	\label{fig:downsampling_app}
\end{figure}

\newpage

\end{document}